\documentclass[journal,twoside,web,xcdraw,svgnames]{ieeecolor}
\usepackage{generic}
\usepackage{cite}
\usepackage{amsmath,amssymb,amsfonts}

\usepackage{hyperref}

\usepackage[T1]{fontenc}
\usepackage[utf8]{inputenc}
\usepackage{graphicx}
\usepackage{amsmath,systeme,amssymb,amsfonts}
\usepackage[version=4]{mhchem}
\usepackage{siunitx}
\usepackage{longtable,tabularx}
\usepackage{comment}
\usepackage{textgreek}
\usepackage{float}
\usepackage[utf8]{inputenc}
\usepackage{tikz, pgfplots}
\usepackage{textcomp}
\usepackage{tcolorbox} 
\usepackage{graphicx}
\usepackage{amsmath}

\usepackage[version=4]{mhchem}
\usepackage{siunitx}
\usepackage{longtable,tabularx}
\usetikzlibrary{positioning, shapes, arrows}
\usepackage{thmtools}
\usepackage{textcomp}
\usepackage{soul}
\usepackage{float,amsfonts,amsthm,color}
\usepackage{amssymb}
\usepackage{graphics} 
\usepackage{booktabs} 
\usepackage{caption, subcaption}
\usepackage{breqn}
\usepackage{array, mathtools}
\usepackage{pgfplots}
\usepackage{siunitx}
\usepackage{algorithm}
\usepackage[noend]{algpseudocode}
\usepackage{tabulary}
\usetikzlibrary{arrows.meta}
\usepackage{stmaryrd} 
\usepackage{varwidth}
\usepackage{comment}
\usetikzlibrary{positioning}

\newtheorem{theorem}{Theorem}
\newtheorem{definition}{Definition}

\newtheorem{lemma}{Lemma}
\newtheorem{remark}{Remark}

\newtheorem{example}{Example}

\def\BibTeX{{\rm B\kern-.05em{\sc i\kern-.025em b}\kern-.08em
    T\kern-.1667em\lower.7ex\hbox{E}\kern-.125emX}}
\begin{document}
\title{Convergence of Payoff-Based Higher-Order Replicator Dynamics in Contractive Games}
\author{Hassan Abdelraouf, Vijay Gupta, and Jeff S. Shamma 
\thanks{Hassan Abdelraouf is with the School of Electrical and Computer Engineering, Purdue University, USA. (e-mail: abdelra5@purdue.edu) }
\thanks{Vijay Gupta is with the School of Electrical and Computer Engineering, Purdue University, USA. (e-mail: gupta869@purdue.edu) }
\thanks{Jeff S. Shamma is with the Department of Industrial and Enterprise Systems Engineering, Grainger College of Engineering, University of Illinois at Urbana-Champaign,  USA. (e-mail: jshamma@illinois.edu) }}
\maketitle

\begin{abstract}
We study the convergence properties of a payoff-based higher-order version of replicator dynamics, a widely studied model in evolutionary dynamics and game-theoretic learning, in contractive games. Recent work has introduced a control-theoretic perspective for analyzing the convergence of learning dynamics through passivity theory, leading to a classification of learning dynamics based on the passivity notion they satisfy, such as \textdelta-passivity, equilibrium-independent passivity, and incremental passivity. We leverage this framework for the study of higher-order replicator dynamics for contractive games, which form the complement of passive learning dynamics. Standard replicator dynamics can be represented as a cascade interconnection between an integrator and the softmax mapping. Payoff-based higher-order replicator dynamics include a linear time-invariant (LTI) system in parallel with the existing integrator. First, we show that if this added system is strictly passive and asymptotically stable, then the resulting learning dynamics converge locally to the Nash equilibrium in contractive games. Second, we establish global convergence properties using incremental stability analysis for the special case of symmetric matrix contractive games.

\end{abstract}

\begin{IEEEkeywords} 
Replicator dynamics, passivity, incremental stability, learning in games.

\end{IEEEkeywords}
\section{Introduction}
\label{sec:introduction}

Population games provide a mathematical framework for modeling strategic interactions in large populations of agents, where the payoff of each strategy depends on the overall distribution of strategies in the population \cite{sandholm2010population}. The population state is represented by a point in the probability simplex whose components denote the fractions of agents selecting each strategy. A payoff function assigns rewards to these strategies based on the current population composition, while evolutionary or learning dynamics describe how agents revise their strategies in response to observed payoffs. The interaction between the learning dynamics, which map payoffs to strategies, and the game, which maps strategies to payoffs, naturally forms a dynamical system whose trajectories characterize the evolution of the population state. A central question in this framework is whether these dynamics converge to equilibrium concepts such as Nash equilibria.

This interaction can be interpreted as a feedback interconnection between the learning dynamics and the game, enabling the use of control-theoretic tools to analyze convergence properties. In particular, passivity-based methods have been used to study convergence of learning dynamics in population games. In \cite{fox2013population}, the notion of $\delta$-passivity was introduced to analyze convergence in contractive games. However, standard replicator dynamics (RD) fails to converge in certain games such as zero-sum games. To address this issue, \cite{gao2020passivity} introduced the exponential RD and showed that it satisfies equilibrium-independent passivity. 

Building on these ideas, \cite{abdelraouf2025passivity} developed a passivity-based classification of learning dynamics based on the passivity notions they satisfy, including $\delta$-passivity, equilibrium-independent passivity, and incremental passivity, and used this framework to study convergence in contractive games. In a related direction, \cite{martins2025counterclockwise} employed counterclockwise dissipativity to analyze convergence in potential games. These works focus primarily on standard (first-order) learning dynamics.  Higher-order RD were studied in \cite{laraki2013higher}, which develops a general framework for higher-order evolutionary dynamics in games. Higher-order variants of RD have also been proposed to improve convergence properties. In particular, \cite{arslan2006anticipatory} introduced anticipatory RD, and \cite{gao2023second} proposed strategic higher-order RD to improve convergence in zero-sum games.

Motivated by these developments, this paper studies the convergence properties of payoff-based higher-order replicator dynamics in contractive games. The main contributions of this paper are twofold. First, we show that payoff-based higher-order replicator dynamics—characterized by the cascade interconnection between the LTI system $((1/s)+h(s))I_n$ and the softmax mapping—converge locally to the Nash equilibrium in contractive games when $h(s)$ is strictly passive. Second, using incremental stability analysis, we study a broader class of learning dynamics characterized by the cascade interconnection between $G(s)I_n$ and the softmax mapping in symmetric matrix contractive games. We show that when these learning dynamics are implemented in a strictly contractive game, the feedback interconnection between the learning dynamics and the game is incrementally asymptotically stable if $G(s)$ is passive and exponentially incrementally stable if $G(s)$ is strictly passive. Under the Nash stationarity assumption, this implies asymptotic convergence to the Nash equilibrium in the first case and exponential convergence in the second.

\section{Notations}
Let $\mathbb{R}_+$ denote the set of nonnegative real numbers. For a vector 
$x \in \mathbb{R}^n$, $x_i$ denotes its $i$-th component and 
$\mathrm{diag}(x)$ denotes the diagonal matrix whose diagonal entries are given by $x$. 
The identity matrix of dimension $n$ is denoted by $I_n$. 
The vectors of all ones and all zeros in $\mathbb{R}^n$ are denoted by  $\mathbf{1}_n$ and $\mathbf{0}_n$, respectively. 
For $x \in \mathbb{R}^n$, $\|x\|$ denotes the Euclidean norm. For matrices $A$ and $B$, $A\otimes B$ denotes their Kronecker product.
The probability simplex in $\mathbb{R}^n$ is defined as
\(
\Delta_n = \{x \in \mathbb{R}^n : x_i \ge 0,\; i=1,\dots,n,\; \mathbf{1}_n^\top x = 1 \}.
\)
Its interior is denoted by $\mathrm{Int}(\Delta_n)$. The tangent space of $\text{Int}(\Delta_n)$ is
\(
\mathcal{Z} = \{z \in \mathbb{R}^n : \mathbf{1}_n^\top z = 0\}.
\)
Let $N \in \mathbb{R}^{n \times (n-1)}$ be a matrix whose columns form an orthonormal basis of $\mathcal{Z}$, i.e.,
$N^\top N = I_{n-1}$ and $\mathbf{1}_n^\top N = 0$.
The softmax mapping $\sigma:\mathbb{R}^n \to \mathrm{Int}(\Delta_n)$ is defined by
\[
\sigma(x)_i = \frac{\exp(x_i)}{\sum_{j=1}^n \exp(x_j)}
\]

\section{Preliminaries}

\subsection{Replicator dynamics and its higher-order variants}

Replicator dynamics (RD) is a widely studied evolutionary dynamics in population games and learning in games \cite{taylor1978evolutionary, mertikopoulos2016learning}. In RD, the dynamics are given by
\begin{equation}
\label{eq:RD}
\dot z = p, \qquad x=\sigma(z),
\end{equation}
where $p(t)\in\mathbb{R}^n$ is the payoff vector, $z\in\mathbb{R}^n$ is the score vector, and $x\in\Delta_n$ is the resulting mixed strategy. Thus, RD can be viewed as the cascade interconnection between the integrator $(1/s)I_n$ and the softmax mapping $\sigma(\cdot)$.

Motivated by this representation, we consider payoff-based higher-order variants of RD obtained by augmenting the score dynamics with an additional linear time-invariant (LTI) system driven by the payoff signal. Specifically, we consider dynamics of the form
\begin{equation}
\label{eq:predictive_RD}
\begin{aligned}
\dot r &= p,\\
\dot x_h &= (A_h\otimes I_n)x_h + (B_h\otimes I_n)p,\\
m &= (C_h\otimes I_n)x_h,\\
z &= r + m,\\
x &= \sigma(z),
\end{aligned}
\end{equation}
where $A_h\in\mathbb{R}^{m\times m}$, $B_h\in\mathbb{R}^{m\times1}$, and $C_h\in\mathbb{R}^{1\times m}$ define a realization of
\(
h(s)=C_h(sI_m-A_h)^{-1}B_h.
\).  Here $r\in\mathbb{R}^n$ is the standard RD score, $x_h\in\mathbb{R}^{mn}$ is the internal state of the added LTI system, and $m\in\mathbb{R}^n$ is the resulting prediction signal. Equivalently, \eqref{eq:predictive_RD} is the cascade interconnection of
\(
G(s)=\left((1/s)+h(s)\right)I_n
\)
with the softmax mapping $\sigma(\cdot)$. This class includes anticipatory variants of RD \cite{arslan2006anticipatory}. We refer to these dynamics as \emph{payoff-based} because the additional LTI subsystem is driven by the payoff signal. 
This differs from the strategic higher-order RD introduced in \cite{gao2023second}, where the additional subsystem is driven by the strategy signal.

\subsection{Population Games and Nash Stationarity}

We consider single-population games \cite{sandholm2010population}, where the population state is $x\in\Delta_n$ and the payoff mapping is given by
\(
p=F(x), \quad F:\Delta_n\to\mathbb{R}^n.
\)

\begin{definition}[Nash Equilibrium]
A state $x^*\in\Delta_n$ is a Nash equilibrium of the population game
$F$ if
\[
x^{*\top}F(x^*) \ge z^\top F(x^*) \quad \forall z\in\Delta_n.
\]
The set of Nash equilibria is denoted by $\mathrm{NE}(F)$.
\end{definition}

In this paper, we focus on contractive games.
\begin{definition}[Contractive Games \cite{hofbauer2009stable}]
A population game $F$ is called contractive if
\(
(x-y)^\top\big(F(x)-F(y)\big)\le 0
\quad \forall x,y\in\Delta_n.
\)
If equality holds only when $x=y$, the game is said to be
\emph{strictly contractive}.
\end{definition}

Contractive games are also referred to as stable games
\cite{hofbauer2009stable}. When $F$ is continuously differentiable,
contractiveness admits the following differential characterization.

\begin{theorem}[Stable Games \cite{hofbauer2009stable}]
If $F$ is continuously differentiable, then $F$ is contractive if and only if its Jacobian satisfies
\(
z^\top \nabla F(x)z \le 0
\quad \forall x\in\Delta_n,\; z\in \mathcal{Z}.
\)
\end{theorem}
 \noindent In a strictly contractive game, the Nash equilibrium is unique \cite{sandholm2015population}.
 
\begin{definition}[Nash Stationarity \cite{sandholm2010population}]
A learning dynamics model $\dot{x}=\mathcal{V}(x,F(x))$
satisfies \emph{Nash stationarity} if
\[
\mathcal{V}(x^*,F(x^*))=0
\quad \iff \quad
x^* \in \mathrm{NE}(F).
\]
\end{definition}

Standard RD satisfies Nash stationarity in the interior of the simplex \cite{sandholm2009pairwise}. In particular, if $x^*\in\mathrm{Int}(\Delta_n)$ is a Nash equilibrium of the population game $F$, then
\(
F(x^*) \in \mathrm{span}(\mathbf{1}_n)
\)
\cite{sandholm2010population}. 
The payoff-based higher-order RD \eqref{eq:predictive_RD} also satisfies Nash stationarity on $\mathrm{Int}(\Delta_n)$ when the transfer function $h(s)$ is asymptotically stable, as stated in the following lemma.

\begin{lemma} \label{lemma:Nash_stationarity_higher_order}
Consider the payoff-based higher-order replicator dynamics \eqref{eq:predictive_RD}. 
If the transfer function $h(s)=C_h(sI_m-A_h)^{-1}B_h$ is asymptotically stable, then the induced strategy dynamics satisfy Nash stationarity on $\mathrm{Int}(\Delta_n)$.
\end{lemma}
The proof is straightforward and omitted for brevity.

\subsection{Passivity and incremental stability}
In this subsection, we recall the notions of passivity and incremental stability used throughout the paper.

Consider the nonlinear system
\begin{equation}
\label{eq:state_space_model}
\begin{aligned}
\dot{x} &= f(x,u),\\
y &= h(x,u),
\end{aligned}
\end{equation}
where $x\in\mathbb{R}^n$ is the state, $u\in\mathbb{R}^m$ is the input, and $y\in\mathbb{R}^m$ is the output.

\begin{definition}[Passivity]
The system \eqref{eq:state_space_model} is passive  if there exists a continuously differentiable storage function $V:\mathbb{R}^n\rightarrow\mathbb{R}_+$ such that
\(
\dot V(x)\le u^\top y
\)
along all trajectories.
\end{definition}

For an LTI system
\(
\dot{x}=Ax+Bu,\quad y=Cx+Du,
\)
with transfer function $H(s)=C(sI-A)^{-1}B+D$, passivity is equivalent to
\(
H(j\omega)+H(j\omega)^*\succeq0
\quad \forall\,\omega\in\mathbb{R}.
\)
The system is strictly passive if
\(
H(j\omega)+H(j\omega)^*\succeq\delta I
\quad \forall\,\omega\in\mathbb{R}
\)
for some $\delta>0$. In the SISO case, passivity condition reduces to 
\(
\operatorname{Re}\{H(j\omega)\}\ge0,
\)
while strict passivity corresponds to
\(
\operatorname{Re}\{H(j\omega)\} > 0
\)
for all $\omega$.

\begin{definition}[Incremental Stability \cite{forni2013differential}]
Consider the autonomous system
\(
\dot{x}=f(x), \; x\in\mathcal{C},
\)
where $\mathcal{C}\subseteq\mathbb{R}^n$ is a forward-invariant  set. The system is \emph{incrementally stable} if there exists a class $\mathcal{K}$ function $\alpha$ such that for any two solutions $x(t;x_0)$ and $y(t;y_0)$ with initial conditions $x_0,y_0\in\mathcal{C}$,
\[
\|x(t;x_0)-y(t;y_0)\|
\le
\alpha(\|x_0-y_0\|).
\]

It is \emph{incrementally asymptotically stable} if it is incrementally stable and
\[
\lim_{t\to\infty}\|x(t;x_0)-y(t;y_0)\|=0 .
\]

It is \emph{exponentially incrementally stable} if there exist constants $M>0$ and $\lambda>0$ such that
\[
\|x(t;x_0)-y(t;y_0)\|
\le
M e^{-\lambda t}\|x_0-y_0\|.
\]
\end{definition}
Since the learning dynamics studied in this paper are characterized by a cascade interconnection between an LTI system and the softmax mapping, we recall two properties of $\sigma$ used in the analysis \cite{gao2017properties}. Its Jacobian is
\(
\nabla\sigma(v)=\mathrm{diag}(\sigma(v))-\sigma(v)\sigma(v)^\top ,
\)
which is symmetric positive semidefinite, with
\(
w^\top\nabla\sigma(v)w\ge0
\)
for all $w\in\mathbb{R}^n$, and equality if and only if $w\in\mathrm{span}\{\mathbf{1}_n\}$. Moreover,
\(
\sigma(v+c\,\mathbf{1}_n)=\sigma(v),\; \forall c\in\mathbb{R}.
\)

The following lemma will be used in the convergence analysis for matrix contractive games.

\begin{lemma}\label{lem:restriction_PD}
Let $N\in\mathbb{R}^{n\times(n-1)}$ be a matrix whose columns form an orthonormal basis of the tangent space $\mathcal{Z}$.
Let $F\in\mathbb{R}^{n\times n}$ be symmetric.
\begin{enumerate}
\item If $z^\top Fz>0$ for all $z\in \mathcal{Z}\setminus\{0\}$,
then $N^\top FN\succ0$.
\item If $z^\top Fz\ge0$ for all $z\in \mathcal{Z}$,
then $N^\top FN\succeq0$.
\end{enumerate}
\end{lemma}

\begin{proof}
Let $y\in\mathbb{R}^{n-1}$ and define $z=Ny$.
Since the columns of $N$ span $\mathcal{Z}$, we have $z\in \mathcal{Z}$ and
$y\neq0$ implies $z\neq0$.
Then
\(
y^\top(N^\top FN)y
=
(Ny)^\top F(Ny)
=
z^\top Fz .
\)
The result follows directly from the definiteness of $F$ on $\mathcal{Z}$.
\end{proof}

\section{Local Convergence of Higher-Order Replicator Dynamics in Contractive Games}

We study the local stability of payoff-based higher-order replicator dynamics when implemented in contractive  games. The interaction between the learning dynamics and the game forms a feedback interconnection, where the learning dynamics maps payoffs to strategies and the game maps strategies to payoffs.

\begin{theorem}
Consider the payoff-based higher-order replicator dynamics \eqref{eq:predictive_RD} implemented in a contractive game. Let $x^* \in \mathrm{Int}(\Delta_n)$ be an isolated mixed-strategy Nash equilibrium. If the transfer function
\(
h(s)=C_h(sI_m-A_h)^{-1}B_h
\)
is strictly passive, then $x^*$ is locally asymptotically stable.
\end{theorem}
\begin{proof}
Consider the payoff-based higher-order replicator dynamics
\eqref{eq:predictive_RD} interconnected with the  game
$p = F(x)$. Since the dynamics satisfies Nash stationarity in
$\mathrm{Int}(\Delta_n)$, any Nash equilibrium $x^*\in\mathrm{Int}(\Delta_n)$ satisfies \( p^* = F(x^*) = \alpha \mathbf{1}_n \)
for some $\alpha\in\mathbb{R}$ \cite{sandholm2010population}.
Using the shift-invariance property of the softmax mapping, the
equilibrium signals can be written as
\(
r^*(t) = r_0 + \alpha t\,\mathbf{1}_n,
\;
x_h^* = -\alpha(A_h^{-1}B_h)\otimes \mathbf{1}_n,
\)
\(
m^* = -\alpha(C_hA_h^{-1}B_h)\mathbf{1}_n,
\;
x^* = \sigma(r_0).
\)
Choosing $r(0)=\log x^* + c\mathbf{1}_n$ for some $c\in\mathbb{R}$
ensures $x^*=\sigma(r_0)$. Linearizing the closed-loop dynamics around the equilibrium trajectroy yields
\[
\begin{aligned}
\dot{\Delta r} &= \Delta p,\\
\dot{\Delta x_h} &= (A_h\otimes I_n)\Delta x_h + (B_h\otimes I_n)\Delta p,\\
\Delta m &= (C_h\otimes I_n)\Delta x_h,\\
\Delta x &= Q(\Delta r+\Delta m),\\
\Delta p &= \nabla F(x^*)\,\Delta x,
\end{aligned}
\]
where $Q = \nabla\sigma(r_0) = \mathrm{diag}(x^*) - x^*{x^*}^\top$.
Admissible deviations of the strategy lie in the tangent space
$\mathcal{Z}$. Let
\(
\Delta r = N\delta r,\quad
\Delta p = N\delta p,\quad
\Delta x = N\delta x,\quad
\Delta m = N\delta m,\quad
\Delta x_h = \mathcal N \delta x_h,
\)
where $\mathcal N = I_m\otimes N$.
Substituting into the linearized dynamics gives
\begin{equation*}
\begin{aligned}
\dot{\delta r} &= \delta p, &
\tilde A_h &= A_h \otimes I_{n-1},\\
\dot{\delta x_h} &= \tilde A_h \delta x_h + \tilde B_h \delta p, &
\tilde B_h &= B_h \otimes I_{n-1},\\
\delta m &= \tilde C_h \delta x_h, &
\tilde C_h &= C_h \otimes I_{n-1},\\
\delta x &= \tilde Q(\delta r+\delta m), &
\tilde Q &= N^\top Q N,\\
\delta p &= \tilde F \delta x, &
\tilde F &= N^\top \nabla F(x^*) N .
\end{aligned}
\end{equation*}
Since $Q\succeq0$ and $\ker(Q)=\mathrm{span}(\mathbf{1}_n)$,
$Q$ is positive definite on $\mathcal{Z}$. Hence, by Lemma~\ref{lem:restriction_PD}, $\tilde Q \succ0$. Define
\(
\delta x_I = \tilde Q\,\delta r, \quad
\delta x_P = \tilde Q\,\delta m,
\)
so that $\delta x = \delta x_I + \delta x_P$.
The resulting linearized dynamics can be interpreted as a feedback interconnection between the game $\tilde F$ and a linear system composed of the parallel connection of the integrator subsystem and the auxiliary LTI subsystem (see Fig.~\ref{fig:linearized_local}). 
\begin{figure}[H]
    \centering
    \includegraphics[width=0.4\linewidth]{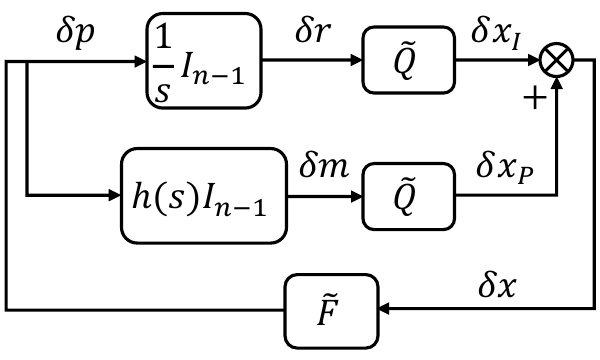}
    \caption{Block diagram of the linearized local dynamics.}
    \label{fig:linearized_local}
\end{figure}

The  auxiliary subsystem has transfer function \( H_Q(s) = h(s)\tilde Q .\)
Since $h(s)$ is strictly passive and $\tilde Q\succ0$,
\(
H_Q(j\omega)+H_Q(j\omega)^*
= 2\,\mathrm{Re}(h(j\omega))\,\tilde Q .
\)
Since $\mathrm{Re}(h(j\omega))\ge\alpha>0$, $\lambda_{\min}(H_{\mathrm{Q}}(j\omega)+H_{\mathrm{Q}}^*(j\omega)) \geq 2 \alpha \lambda_{\min}(\tilde{Q})= \delta>0$. Thus the auxiliary subsystem is strictly passive. By the KYP lemma \cite[Lemma~6.3]{khalil2002nonlinear}, there exist matrices $P\succ0$, $L$, and $\epsilon>0$ satisfying
\(
P\tilde A_h + \tilde A_h^\top P
= -L^\top L - \epsilon P,
\quad
P\tilde B_h = \tilde C_h^\top \tilde Q .
\)

Consider the Lyapunov function candidate
\[
V(\delta r,\delta x_h)
= \tfrac12\,\delta r^\top \tilde Q\,\delta r
+ \tfrac12\,\delta x_h^\top P\,\delta x_h .
\]
Its derivative along trajectories is
\[
\begin{aligned}
\dot V
&= \delta p^\top \tilde Q\,\delta r
+ \delta x_h^\top (P\tilde A_h+\tilde A_h^\top P)\delta x_h
+ \delta x_h^\top \tilde C_h^\top \tilde Q\,\delta p \\
&= \delta p^\top \delta x
- \delta x_h^\top L^\top L\,\delta x_h
- \epsilon\,\delta x_h^\top P\,\delta x_h .
\end{aligned}
\]
Using $\delta p=\tilde F\delta x$ and contractiveness of the game,
$\delta x^\top \tilde F \delta x \le 0$, which yields
\(
\dot V \le -\epsilon\,\delta x_h^\top P\,\delta x_h \le 0 .
\)
Hence $\dot V=0$ only when $\delta x_h=\mathbf{0}_{m(n-1)}$. By LaSalle's invariance principle, the largest invariant set is $(\delta r,\mathbf 0_{m(n-1)})$. On this set, $\delta p=\mathbf{0}_{n-1}$ and $\delta x=\mathbf{0}_{n-1}$, which implies $\delta r=\mathbf{0}_{n-1}$. Therefore $x^*$ is locally asymptotically stable.
\end{proof}
\begin{example}
Consider the Rock–Paper–Scissors (RPS) game, 
\[
F=\begin{bmatrix}
0 & -1 & 1\\
1 & 0 & -1\\
-1 & 1 & 0
\end{bmatrix}.
\]
The game has a unique Nash equilibrium $x^*=(1/3,1/3,1/3)$. Moreover, it is lossless since $z^\top F z=0$ for all $z\in \mathcal{Z}$. It is well known that the standard replicator dynamics fails to converge in the RPS game \cite{arslan2006anticipatory}. In contrast, the payoff-based higher-order replicator dynamics~\eqref{eq:predictive_RD} converges locally to the Nash equilibrium for any strictly passive $h(s)$. For example, with $h(s)=\tfrac{2s+3}{s^2+3s+2}$, Fig.~\ref{fig:convergence_RPS} shows that the Nash equilibrium is locally asymptotically stable.
\end{example}

\begin{figure}[H]
\centering
\includegraphics[width=0.7\linewidth]{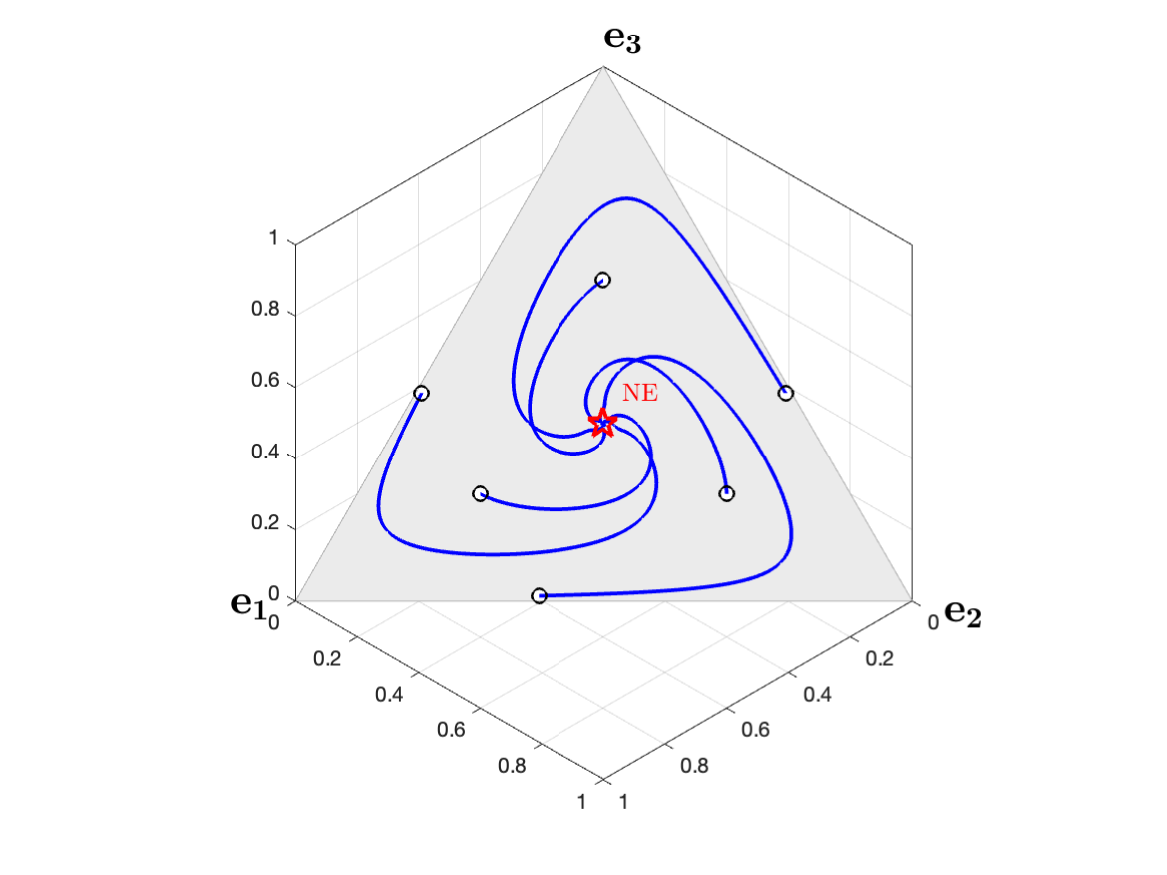}
\caption{Local convergence of payoff-based higher-order replicator dynamics for $h(s)=\tfrac{2s+3}{s^2+3s+2}$ in the RPS game.}
\label{fig:convergence_RPS}
\end{figure}

\section{Global Convergence of Higher-order replicator dynamics in matrix contractive games}

In this section, we study the convergence properties of a broad class of learning dynamics characterized by a cascade interconnection between the diagonal transfer function matrix $G(s)=g(s)I_n$, where $g(s)$ is passive, and the softmax mapping $\sigma(\cdot)$. This class includes several learning models such as the payoff-based higher-order replicator dynamics \eqref{eq:predictive_RD}, where $g(s)=(1/s)+h(s)$, replicator dynamics under perturbed or delayed payoff observations, and exponential replicator dynamics
\cite{gao2020passivity}.

We first establish a general result that relates the stability of the
variational dynamics of a nonlinear system to incremental stability of the nonlinear flow. In particular, we show that if a common quadratic Lyapunov function certifies stability of all trajectory linearizations, then the nonlinear system is globally incrementally stable.

\begin{theorem}\label{thm:incremental_stability}
Consider the nonlinear system
\(
\dot x = f(x),
\; x\in\mathbb{R}^n,
\)
with Jacobian $A(x)=\nabla f(x)$. Suppose there exists a symmetric matrix $P\succ0$ such that one of the following holds:
\begin{enumerate}
\item \textbf{Exponential case.}
If there exists $Q\succ0$ satisfying
\begin{equation}
PA(x)+A(x)^\top P \preceq -Q,
\qquad \forall x\in\mathbb{R}^n,
\label{eq:contraction_cond_exp}
\end{equation}
then the system is globally exponentially incrementally stable.

\item \textbf{Asymptotic case.}
If there exists a continuous matrix function
$Q:\mathbb{R}^n\to\mathbb{S}_+^n$ such that
\begin{equation}
PA(x)+A(x)^\top P \preceq -Q(x),
\qquad \forall x\in\mathbb{R}^n,
\label{eq:contraction_cond_asym}
\end{equation}
then the system is globally incrementally asymptotically stable, provided the largest invariant set contained in
\[
\mathcal{E}
=
\Big\{(y,e):
e^\top\!\Big(\int_0^1 Q(y+\theta e)\,d\theta\Big)e=0
\Big\}
\]
is $\{(y,e):e=0\}$.
\end{enumerate}
If an equilibrium $x^\star$ exists, then it inherits the corresponding global stability property.
\end{theorem}

\begin{proof}
\textbf{(1) Exponential case.}
If \eqref{eq:contraction_cond_exp} holds, then the system is contracting with respect to the constant metric $P\succ0$. Hence global exponential incremental stability follows from Theorem 2 in \cite{lohmiller1998contraction}. Moreover, if an equilibrium $x^\star$ exists, then setting $y(t)\equiv x^\star$, which is a solution of the system, implies global exponential stability of $x^\star$.

\textbf{(2) Asymptotic case.}
Let $x(\cdot)$ and $y(\cdot)$ be two solutions and define the error
\(
e=x-y.
\)
Using the mean-value form of the Jacobian,
\[
\dot e
=
f(x)-f(y)
=
\widehat A(t)e,
\qquad
\widehat A(t)
=
\int_0^1 A(y+\theta e)\,d\theta .
\]
Consider the Lyapunov function
\(
V(e)=e^\top P e,
\)
where $P\succ0$. Then
\begin{align*}
    \dot V &=
e^\top\!\big(P\widehat A+\widehat A^\top P\big)e\\
&=
e^\top\!\left(\int_0^1
\big(PA(y+\theta e)+A(y+\theta e)^\top P\big)\,d\theta\right)e .
\end{align*}
By \eqref{eq:contraction_cond_asym},
\[
P\widehat A+\widehat A^\top P
\preceq
-\int_0^1 Q(y+\theta e)\,d\theta .
\]
Hence
\[
\dot V
\le
-
e^\top\!\left(\int_0^1 Q(y+\theta e)\,d\theta\right)e
\le 0 .
\]
Therefore \(V\) is nonincreasing. The set where \(\dot V=0\) is
\[
\mathcal{E}
=
\Big\{(y,e):
e^\top\!\Big(\int_0^1 Q(y+\theta e)\,d\theta\Big)e=0
\Big\}.
\]
By LaSalle's invariance principle, trajectories converge to the largest invariant subset of \(\mathcal{E}\). By assumption, this invariant set is \(\{(y,e):e=0\}\), so \(e(t)\to0\). Thus the system is globally incrementally asymptotically stable. Finally, if an equilibrium \(x^\star\) exists, taking \(y(t)\equiv x^\star\) gives global asymptotic stability of \(x^\star\).
\end{proof}

Theorem~\ref{thm:incremental_stability} characterizes incremental
stability of a nonlinear system in terms of a common quadratic Lyapunov function for all trajectory linearizations. We now apply this result to establish global convergence of the learning dynamics obtained by connecting $G(s)I_n$ with $\sigma(\cdot)$ in symmetric matrix contractive games (see Fig.~\ref{fig:implementation_in_matrix_game}).

\begin{figure}[H]
    \centering
    \includegraphics[width=0.4\linewidth]{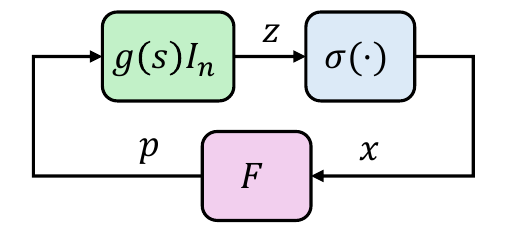}
    \caption{Implementation of generalized RD in a matrix game.}
    \label{fig:implementation_in_matrix_game}
\end{figure}

\begin{theorem}\label{thrm:general_LD_convergence}
Consider the learning dynamics obtained from the cascade
interconnection of a transfer function
\(
G(s)=g(s)I_n
\)
with the softmax mapping $\sigma(\cdot)$, and suppose the dynamics
satisfy Nash stationarity. When implemented in a symmetric matrix
strictly contractive game
\[
p = Fx, \qquad F = F^\top, \qquad
z^\top F z < 0 \;\; \forall z \in \mathcal{Z} \setminus \{0\},
\]
the following hold:
\begin{enumerate}
\item If $g(s)$ is passive, the unique Nash equilibrium is  globally asymptotically stable.

\item If $g(s)$ is strictly passive, the the unique NE is globally exponentially stable.
\end{enumerate}
\end{theorem}

\begin{proof}
Let $g(s)$ admit a state--space realization $(A,B,C,D)$ with
$A\in\mathbb{R}^{m\times m}$,
$B\in\mathbb{R}^{m\times 1}$,
$C\in\mathbb{R}^{1\times m}$.
We take $D=0$ to avoid direct feedthrough from $p$ to $z$.\footnote{If $D\neq0$, the output $z$ depends instantaneously on $p$, creating an algebraic loop because both $F$ and $\sigma(\cdot)$ are static maps.}.  The dynamics of $G(s)$ are
\begin{align*}
\dot{\xi} &= (A\otimes I_n)\xi + (B\otimes I_n)p,\\
z &= (C\otimes I_n)\xi .
\end{align*}
The closed-loop system is therefore
\begin{align*}
\dot{\xi} &= (A\otimes I_n)\xi + (B\otimes I_n)p,\\
z &= (C\otimes I_n)\xi,\\
x &= \sigma(z),\\
p &= Fx .
\end{align*}

Linearizing around a trajectory
$(\bar\xi,\bar z,\bar x,\bar p)$ yields
\begin{align*}
\Delta \dot{\xi} &= (A\otimes I_n)\Delta\xi + (B\otimes I_n)\Delta p,\\
\Delta z &= (C\otimes I_n)\Delta\xi,\\
\Delta x &= Q(\bar x)\Delta z,\\
\Delta p &= F\Delta x,
\end{align*}
where 
\(
Q(\bar x)=\operatorname{diag}(\bar x)-\bar x\bar x^\top
\) is evaluated along the trajectory $\bar x(t)$ and is therefore time-varying.  Because admissible strategy deviations lie in the tangent space
$\mathcal{Z}$, write $\Delta x=N\delta x$ where
$N\in\mathbb{R}^{n\times(n-1)}$ satisfies
$N^\top N=I_{n-1}$ and $N^\top\mathbf{1}_n=0$.
Define $\mathcal N=I_m\otimes N$ and
$\delta\xi=\mathcal N^\top\Delta\xi$.
The reduced dynamics become
\begin{equation*}
\begin{aligned}
\delta \dot{\xi} &= \tilde A\,\delta \xi + \tilde B\,\tilde F\,\delta x, &
\tilde A &= A \otimes I_{n-1},\\
\delta z &= \tilde C\,\delta \xi, &
\tilde B &= B \otimes I_{n-1},\\
\delta x &= \tilde Q(\bar x)\,\delta z, &
\tilde C &= C \otimes I_{n-1},
\end{aligned}
\end{equation*}
where
\(
\tilde F = N^\top F N, \) and 
\(
\tilde Q(\bar x)=N^\top Q(\bar x)N .
\)
Since $Q(\bar{x})\succeq0$ and $\ker Q(\bar{x})=\operatorname{span}\{\mathbf{1}_n\}$,  $Q(\bar{x})$ is positive definite on
$\mathcal{Z}$. Hence Lemma~\ref{lem:restriction_PD} implies
\( \tilde Q(\bar x)\succ0 .\)
Similarly, strict contractiveness of $F$ yields \( \tilde F \prec 0 \).  The resulting system can be interpreted as the negative-feedback
interconnection between $H(s)$ and  $\tilde Q(\bar x)$ where 
\(
H(s)=-\tilde C(sI-\tilde A)^{-1}\tilde B\tilde F
      =-g(s)\tilde F 
\) as illustrated in Fig.~\ref{fig:closed_loop_linearized}.
\begin{figure}[H]
    \centering
    \includegraphics[width=0.35\linewidth]{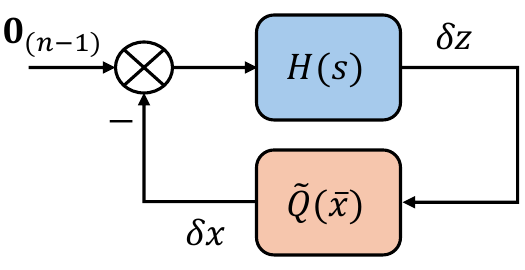}
    \caption{Block diagram representation of the linearized system.}
    \label{fig:closed_loop_linearized}
\end{figure}

Since $F$ is symmetric, it follows that
\[
H(j\omega)+H(j\omega)^*
=
-2\,\mathrm{Re}\big(g(j\omega)\big)\tilde F .
\]

\paragraph{Case 1: $g(s)$ passive.}
If $g(s)$ is passive then
$\mathrm{Re}(g(j\omega))\ge0$.
Because $-\tilde F\succ0$, it follows that
$H(j\omega)+H(j\omega)^*\succeq0$,
so $H(s)$ is passive.
By the positive real lemma {\cite[Lemma~6.2]{khalil2002nonlinear}}  there exists $P\succ0$ satisfying
\[
P\tilde A+\tilde A^\top P\preceq0,
\qquad
-P\tilde B\tilde F=\tilde C^\top .
\]

Consider $V(\delta\xi)=\tfrac12\delta\xi^\top P\delta\xi$.
Along trajectories,
\begin{align*}
\dot V
&=
\tfrac12\delta\xi^\top(P\tilde A+\tilde A^\top P)\delta\xi
-\delta\xi^\top\tilde C^\top\delta x 
\le
-\delta z^\top\tilde Q(\bar x)\delta z .
\end{align*}

Since $\tilde Q(\bar x)\succ0$, $\dot V\le0$ with equality only when
$\tilde C\delta\xi=0$. Hence the set where $\dot V=0$ is
$\ker(\tilde C)$. Because $(A,C)$ is observable, the lifted pair
$(\tilde A,\tilde C)$ is also observable, implying that the largest
invariant subset of $\ker(\tilde C)$ is $\{0\}$. By LaSalle's
invariance principle, all trajectories converge to the origin, so every linearized LTV system is asymptotically stable.
By Theorem~\ref{thm:incremental_stability}, the nonlinear closed-loop
system is globally incrementally asymptotically stable. Under Nash
stationarity and strict contractiveness, the unique Nash equilibrium is therefore globally asymptotically stable.

\paragraph{Case 2: $g(s)$ strictly passive.}
If $g(s)$ is strictly passive then
$\mathrm{Re}(g(j\omega))\ge\alpha>0$,
implying
\[
\lambda_{\min}\!\big(H(j\omega)+H(j\omega)^*\big)
\ge
2\alpha\lambda_{\min}(-\tilde F)>0 .
\]
Thus $H(s)$ is strictly passive.
By the KYP lemma there exists $P\succ0$ such that
\[
P\tilde A+\tilde A^\top P
=
-L^\top L-\varepsilon P,
\qquad
-P\tilde B\tilde F=\tilde C^\top .
\]

With $V(\delta\xi)=\tfrac12\delta\xi^\top P\delta\xi$,
\begin{align*}
\dot V
&=
-\tfrac12\|L\delta\xi\|^2
-\tfrac12\varepsilon\,\delta\xi^\top P\delta\xi
-\delta z^\top\tilde Q(\bar x)\delta z 
\le
-\tfrac{\varepsilon}{2}V .
\end{align*}

Hence all linearizations are exponentially stable.
Applying Theorem~\ref{thm:incremental_stability} yields global
exponential incremental stability of the nonlinear system.
By Nash stationarity and strict contractiveness of the game,
the unique Nash equilibrium is globally exponentially stable.
\end{proof}
\begin{remark}
The asymptotic convergence result also holds when the game is contractive with a unique interior Nash equilibrium and $g(s)$ is strictly passive. 
\end{remark}
\begin{example}[Congestion game \cite{park2021kl}]
\label{ex:congestion_game}
Consider the congestion network shown in Fig.~\ref{fig:congestion_network}
with three routes
$O\!\to\!A\!\to\!D$, $O\!\to\!B\!\to\!D$, and $O\!\to\!A\!\to\!B\!\to\!D$.
Let $x=(x_1,x_2,x_3)\in\Delta_3$ denote the population state, where
$x_i$ is the fraction of agents selecting route $i$.
Each edge has latency equal to its total flow, except edge $A\!\to\!D$,
whose latency is twice the flow. The resulting route costs are
\[
c(x)=
\begin{bmatrix}
3x_1+x_3\\
2x_2+x_3\\
x_1+x_2+3x_3
\end{bmatrix},
\quad
p=-c(x)=\underbrace{-\begin{bmatrix}
3&0&1\\
0&2&1\\
1&1&3
\end{bmatrix}}_{F}x .
\]
Since $F$ is symmetric and negative definite on $\mathcal{Z}$, the game is strictly contractive and admits a unique Nash equilibrium
\[
x^*=\left(\tfrac{4}{11},\,\tfrac{6}{11},\,\tfrac{1}{11}\right).
\]
\end{example}
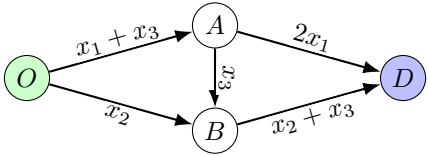
\begin{figure}[H]
  \centering
  \begin{tikzpicture}[
      >=Latex,
      node/.style={circle,draw,minimum size=6mm,inner sep=0pt},
      origin/.style={node,fill=green!20},
      dest/.style={node,fill=blue!25},
      cost/.style={pos=0.5,sloped,allow upside down,fill=white,inner sep=1pt},
      scale=1, every node/.style={transform shape}
    ]
    \node[origin] (O) at (0,0) {$O$};
    \node[node]   (A) at (2.5,0.7) {$A$};
    \node[node]   (B) at (2.5,-0.7) {$B$};
    \node[dest]   (D) at (5,0) {$D$};

    \draw[->,thick] (O) -- node[cost,above] {$x_1+x_3$} (A);
    \draw[->,thick] (O) -- node[cost,below] {$x_2$} (B);
    \draw[->,thick] (A) -- node[cost,above] {$2x_1$} (D);
    \draw[->,thick] (B) -- node[cost,below] {$x_2+x_3$} (D);
    \draw[->,thick] (A) -- node[cost,above] {$x_3$} (B);
  \end{tikzpicture}
  \caption{Congestion network}
  \label{fig:congestion_network}
\end{figure}

\begin{figure}[H]
\centering
\begin{subfigure}{0.48\linewidth}
\centering
\includegraphics[width=\linewidth]{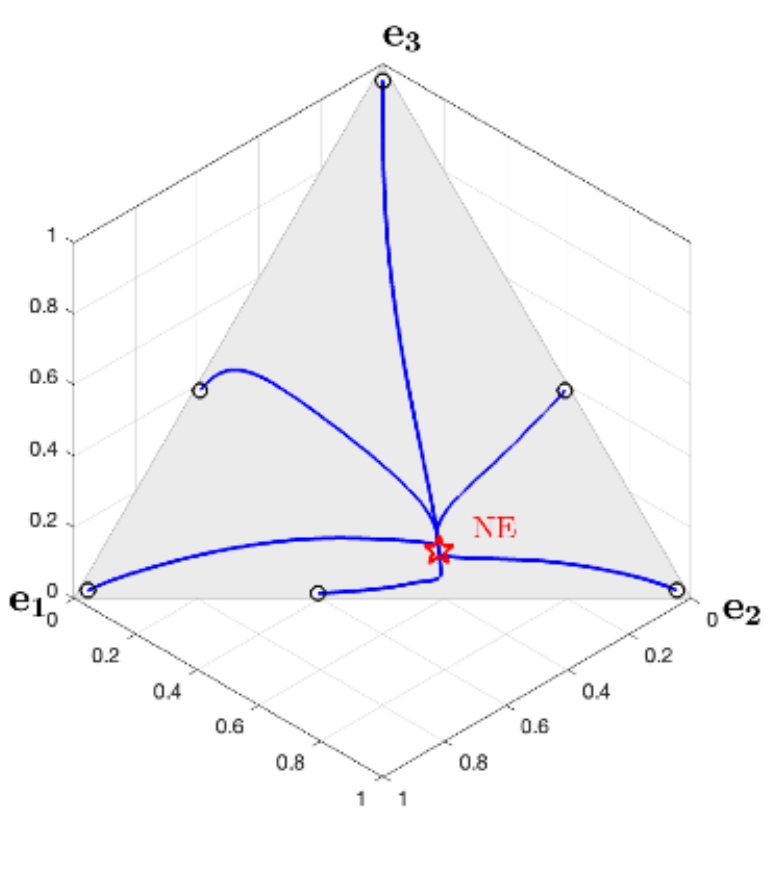}
\caption{\(G(s)=\tfrac{2s^2+3.5s+2}{s^3+3s^2+2s}I_n\).}
\end{subfigure}
\hfill
\begin{subfigure}{0.48\linewidth}
\centering
\includegraphics[width=\linewidth]{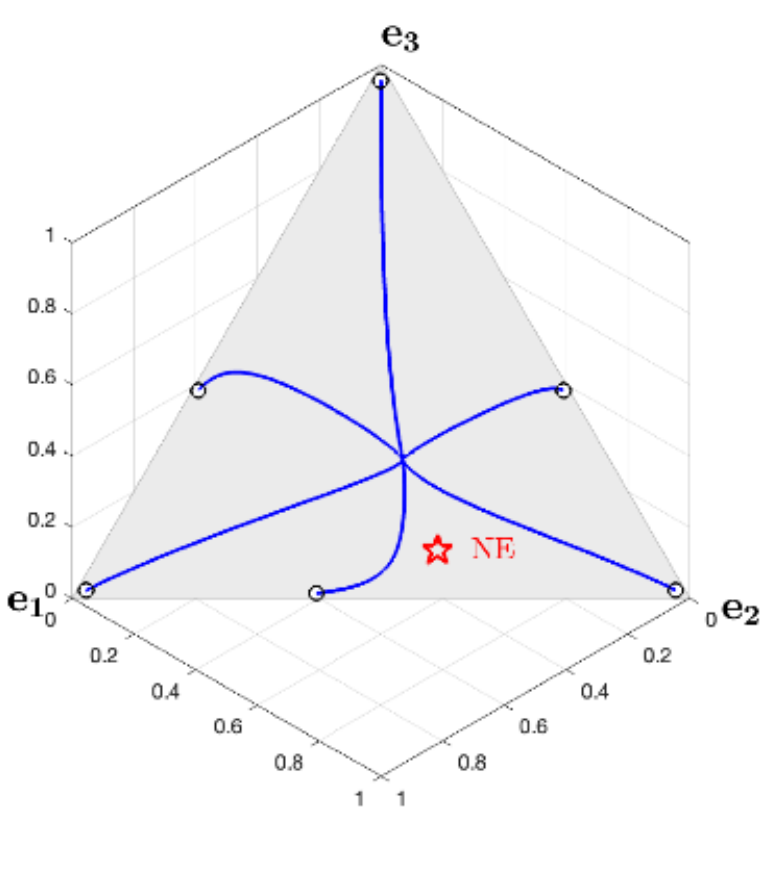}
\caption{ \(G(s)=\tfrac{1}{s+1}I_n\).}
\end{subfigure}

\caption{Convergence of learning dynamics obtained by cascading \(G(s)\) with $\sigma(\cdot)$ in the congestion game of Example~\ref{ex:congestion_game}.}
\label{fig:global_examples}
\end{figure}

Fig.~\ref{fig:global_examples}(a) illustrates convergence of the learning dynamics in the congestion game of Example~\ref{ex:congestion_game}, obtained by cascading \(G(s)=\frac{2s^2+3.5s+2}{s^3+3s^2+2s}I_n\) with $\sigma(\cdot)$. Since $G(s)$ is passive and the dynamics satisfy Nash stationarity on $\mathrm{Int}(\Delta_n)$ by Lemma~\ref{lemma:Nash_stationarity_higher_order}, Theorem~\ref{thrm:general_LD_convergence} implies that the closed-loop dynamics are incrementally asymptotically stable and the unique Nash equilibrium is globally asymptotically stable.

Fig.~\ref{fig:global_examples}(b) shows the exponential replicator dynamics (Ex-RD), corresponding to the cascade interconnection of the strictly passive transfer function \(1/(s+1)I_n\) with $\sigma(\cdot)$. In this case, Theorem~\ref{thrm:general_LD_convergence} implies that the closed-loop dynamics are incrementally exponentially stable. However, since Ex-RD does not satisfy Nash stationarity, the dynamics converge exponentially to another equilibrium \(x^\dagger=\sigma(z^*)\), where \(z^*=F\sigma(z^*)\).

\section{Conclusion}

This paper studied the convergence properties of payoff-based higher-order replicator dynamics in contractive games from a control-theoretic perspective. We first showed that payoff-based higher-order replicator dynamics, characterized by the cascade interconnection between $\left((1/s)+h(s)\right)I_n$ and the softmax mapping, converge locally to the Nash equilibrium when $h(s)$ is strictly passive. We then established global convergence results for a broader class of learning dynamics of the form $G(s)I_n$ interconnected with the softmax mapping in symmetric strictly contractive matrix games. In particular, passivity of $G(s)$ yields asymptotic convergence, while strict passivity yields exponential convergence. Future work will focus on extending the global convergence analysis to contractive games beyond the symmetric matrix case and on studying robustness properties of higher-order replicator dynamics.

\section{References}

\bibliographystyle{ieeetr}
\bibliography{refs.bib}

@inproceedings{abdelraouf2025passivity,
  title={Passivity, no-regret, and convergent learning in contractive games},
  author={Abdelraouf, Hassan and Piliouras, Georgios and Shamma, Jeff S},
  booktitle={2025 IEEE 64th Conference on Decision and Control (CDC)},
  pages={4948--4954},
  year={2025},
  organization={IEEE}
}

@article{mertikopoulos2016learning,
  title={Learning in games via reinforcement and regularization},
  author={Mertikopoulos, Panayotis and Sandholm, William H},
  journal={Mathematics of Operations Research},
  volume={41},
  number={4},
  pages={1297--1324},
  year={2016},
  publisher={INFORMS}
}

@article{gao2023second,
  title={Second-order mirror descent: Convergence in games beyond averaging and discounting},
  author={Gao, Bolin and Pavel, Lacra},
  journal={IEEE Transactions on Automatic Control},
  volume={},
  number={},
  pages={},
  year={2023},
  publisher={IEEE}
}

@article{gao2020passivity,
  title={On passivity, reinforcement learning, and higher order learning in multiagent finite games},
  author={Gao, Bolin and Pavel, Lacra},
  journal={IEEE Transactions on Automatic Control},
  volume={66},
  number={1},
  pages={121--136},
  year={2020},
  publisher={IEEE}
}

@inproceedings{arslan2006anticipatory,
  title={Anticipatory learning in general evolutionary games},
  author={Arslan, Gurdal and Shamma, Jeff S},
  booktitle={Proceedings of the 45th IEEE Conference on Decision and Control},
  pages={6289--6294},
  year={2006},
  organization={IEEE}
}

@article{laraki2013higher,
  title={Higher Order Game Dynamics},
  author={Laraki, Rida and Mertikopoulos, Panayotis},
  journal={Journal of Economic Theory},
  volume={148},
  number={6},
  pages={2666--2695},
  year={2013},
  publisher={Elsevier}
}

@article{fox2013population,
  title={Population Games, Stable Games, and Passivity},
  author={Fox, Michael J and Shamma, Jeff S},
  journal={Games},
  volume={4},
  number={4},
  pages={561--583},
  year={2013},
  publisher={MDPI}
}

@article{sandholm2009pairwise,
  title={Pairwise Comparison Dynamics and Evolutionary Foundations for {Nash} Equilibrium},
  author={Sandholm, William H},
  journal={Games},
  volume={1},
  number={1},
  pages={3--17},
  year={2009},
  publisher={MDPI}
}

@article{hofbauer2009stable,
  title={Stable games and their dynamics},
  author={Hofbauer, Josef and Sandholm, William H},
  journal={Journal of Economic Theory},
  volume={144},
  number={4},
  pages={1665--1693},
  year={2009},
  publisher={Elsevier}
}

@article{gao2017properties,
  title={On the properties of the softmax function with application in game theory and reinforcement learning},
  author={Gao, Bolin and Pavel, Lacra},
  journal={arXiv preprint arXiv:1704.00805},
  year={2017}
}

@book{sandholm2010population,
  title={Population Games and Evolutionary Dynamics},
  author={Sandholm, William H},
  year={2010},
  publisher={MIT press}
}

@incollection{sandholm2015population,
  title={Population games and deterministic evolutionary dynamics},
  author={Sandholm, William H},
  booktitle={Handbook of Game Theory with Economic Applications},
  volume={4},
  pages={703--778},
  year={2015},
  publisher={Elsevier}
}

@book{khalil2002nonlinear,
  title={Nonlinear Systems},
  author={Khalil, Hassan K.},
  edition={3},
  year={2002},
  publisher={Prentice Hall},
  address={Upper Saddle River, NJ}
}

@article{forni2013differential,
  title={A differential {Lyapunov} framework for contraction analysis},
  author={Forni, Fulvio and Sepulchre, Rodolphe},
  journal={IEEE Transactions on Automatic Control},
  volume={59},
  number={3},
  pages={614--628},
  year={2013},
  publisher={IEEE}
}

@article{lohmiller1998contraction,
  title={On contraction analysis for non-linear systems},
  author={Lohmiller, Winfried and Slotine, Jean-Jacques E},
  journal={Automatica},
  volume={34},
  number={6},
  pages={683--696},
  year={1998},
  publisher={Elsevier}
}

@article{taylor1978evolutionary,
  title={Evolutionary stable strategies and game dynamics},
  author={Taylor, Peter D and Jonker, Leo B},
  journal={Mathematical Biosciences},
  volume={40},
  number={1-2},
  pages={145--156},
  year={1978},
  publisher={Elsevier}
}

@inproceedings{park2021kl,
  title={K{L} divergence regularized learning model for multi-agent decision making},
  author={Park, Shinkyu and Leonard, Naomi Ehrich},
  booktitle={2021 American Control Conference (ACC)},
  pages={4509--4514},
  year={2021},
  organization={IEEE}
}

@article{martins2025counterclockwise,
  title={Counterclockwise dissipativity, potential games and evolutionary {N}ash equilibrium learning},
  author={Martins, Nuno C and Cert{\'o}rio, Jair and Hankins, Matthew S},
  journal={IEEE Transactions on Automatic Control},
  year={2025},
  publisher={IEEE}
}

\end{document}